\documentclass[11pt,draftcls,onecolumn]{IEEEtran}

\usepackage{indentfirst,flushend}
\usepackage{amssymb,bm,mathrsfs,bbm,amscd,amsmath,amsthm}
\usepackage{graphicx,pifont}
\usepackage{color}
\usepackage{cases}
\newtheorem{theorem}{Theorem}

\newtheorem{proposition}{Proposition}

\newtheorem{remark}{Remark}
\usepackage{multirow}
\usepackage{cite}
\usepackage{citesort}
\usepackage{multicol}
\usepackage[nooneline,flushleft]{caption2}
\usepackage{algorithm,clrscode}
\usepackage{algorithmic}

\begin{document}

\title{Coordinated Beamforming for Energy Efficient Transmission in Multicell Multiuser Systems}

\author{Shiwen~He, Yongming~Huang,~\IEEEmembership{Member,~IEEE}, Shi~Jin,~\IEEEmembership{Member,~IEEE}, and Luxi~Yang,~\IEEEmembership{Member,~IEEE}
\thanks{S. He, Y. Huang, S. Jin and L. Yang are with the School of Information Science and Engineering, Southeast University, Nanjing 210096, China. (Email:\{hesw01, huangym, jinshi, lxyang\}@seu.edu.cn).}
}

\maketitle

\begin{abstract}
In this paper we study energy efficient joint power allocation and beamforming for coordinated multicell multiuser downlink systems. The considered optimization problem is in a non-convex fractional form and hard to tackle. We propose to first transform the original problem into an equivalent optimization problem in a parametric subtractive form, by which we reach its solution through a two-layer optimization scheme. The outer layer only involves one-dimension search for the energy efficiency parameter which can be addressed using the bi-section search, the key issue lies in the inner layer where a non-fractional sub-problem needs to tackle. By exploiting the relationship between the user rate and the mean square error, we then develop an iterative algorithm to solve it. The convergence of this algorithm is proved and the solution is further derived in closed-form. Our analysis also shows that the proposed algorithm can be implemented in parallel with reasonable complexity. Numerical results illustrate that our algorithm has a fast convergence and achieves near-optimal energy efficiency. It is also observed that at the low transmit power region, our solution almost achieves the optimal sum rate and the optimal energy efficiency simultaneously; while at the middle-high transmit power region, a certain sum rate loss is suffered in order to guarantee the energy efficiency.
\end{abstract}


\section*{\sc \uppercase\expandafter{\romannumeral1}. Introduction}

The increasing demands for better services in wireless communications involve higher transmission rate, lower error rate and enhanced coverage. In order to achieve these objectives, advanced wireless transmission and signal processing techniques have been intensively investigated in the literature~\cite{ConShamai2001,Zhang2004}. Recently energy consumption problem has attracted increasing interest, due to the fact that low energy efficiency in wireless communications will result in high cost for the devices especially the mobile terminals, have negative impact on the environment and even cause health problems~\cite{ConCheng2011}. How to trade off the relationship between the system capacity and the energy consumption has become a key issue for future wireless communications~\cite{MagChen2011,MagXu2011,MagHan2011,ConKumar2011,ConArnold2010}.

It is well known that multiple-input multiple-output (MIMO) technology provides extra degrees of freedom and brings multiplexing and diversity gains. As a result, multi-user MIMO (MU-MIMO) transmission has attracted a lot of research interest in the past few decades and enables significant performance enhancement without additional transmit power and bandwidth resource~\cite{BookTse2005,TWCMarzetta2010}. In particular, it is shown in~\cite{TWCMarzetta2010} that massive MIMO with a large number of antennas equipped at the base station (BS) promises much improved spectral and energy efficiency. In addition, as a powerful tool to mitigate the inter-cell interference resulting from aggressive frequency reuse, BS cooperation, also known as network MIMO or coordinated multi-point transmission and reception (CoMP), has recently received much attention~\cite{JSACGESbert2010,BookGan2013,Dahrouj2010,TSPHuang2011,TWCHuang2012,TWCHuang2013}. An overview of state of the art multicell MIMO cooperation techniques is presented in~\cite{JSACGESbert2010,BookGan2013}. In~\cite{Dahrouj2010}, a coordinated beamforming algorithm is proposed to minimize the transmit power subject to given SINR constraints. Later, with the goal of maximizing the worst user rate, a distributed multicell beamforming
solution is reached which only requires limited intercell coordination~\cite{TSPHuang2011}. By exploiting the property of massive MIMO channels, in~\cite{TWCHuang2013} a distributed coordinated power allocation method is developed to balance the weighted SINR in a multicell massive multiple input single output (MISO) downlink system. Considering the user fairness, a distributed coordinated beamforming scheme is designed to achieve the Pareto boundary of user rate tuples~\cite{TWCHuang2012}, by deriving an approximate uplink-downlink duality. Besides, the sum rate maximization problem for the coordinated beamforming is also widely studied by using the relationship between the user rate and the minimum mean square error (MMSE)~\cite{TWCChristensen2008,ICCJose2011,TSPShi2011,TSPBogale2012} or the branch and bound method~\cite{TSPWeeraddana2011}. In particular, a distributed sum rate maximization solution is achieved for the multicell beamforming system, in which only limited intercell signalling exchange is needed~\cite{TSPWeerad2013}.

Note that these aforementioned references are only concerned with the system throughput or spectral efficiency. Energy efficient system design, which adopts energy efficiency (bit-per-Joule) as the performance metric, has recently drawn much attention in both industry and academia~\cite{TWCVilardeb2010,TCMMiao2010,TWCMiao2011,TWCNg201209,TWCNg201210}. The energy efficiency bound of the relay channel under additive white Gaussian noise is analyzed and computed in~\cite{TWCVilardeb2010}. In~\cite{TCMMiao2010}, a link adaptive transmission method is proposed to maximize energy efficiency by adapting both overall transmit power and its allocation. Later, an energy efficient power optimization scheme is further developed for interference-limited wireless communications~\cite{TWCMiao2011}. In addition, energy efficient resource allocation has been widely studied for the orthogonal frequency division multiple access (OFDMA) downlink systems with a large number of transmit antennas and fixed beamformers or for the multicell OFDMA downlink network with cooperative BSs and single transceiver antenna~\cite{TWCNg201209,TWCNg201210}. More recently, the energy efficient transmission design for massive MIMO systems and small cell networks has become a hot topic due to the potential of significantly improving both the spectral and the energy efficiency~\cite{TCOMNgo2013,ICTBjornson2013}. It is worth mentioning that all these works above only consider simple transceivers where the transmitter is equipped with a single antenna or with a fixed beamformer. The joint optimization of energy efficient power allocation and beamforming is still an open problem.

Motivated by this, in this paper we aim to design an energy efficient transmission for multicell multi-user MISO (MU-MISO) downlink system by jointly optimizing the transmit powers and beamforming vectors. The original problem is non-convex and is difficult to solve directly due to the coupling between variables and its fractional form. We propose to first transform it into an equivalent subtractive-form optimization problem by exploiting the fractional programming~\cite{JstorJagan1966,JstorDink1967,MathCrouzeix1991}. We further reveal that this equivalent problem can be solved using one dimension search method~\cite{BookBoyd2004}, in each search a sub-problem needs to address. Then, we develop an efficient optimization beamforming algorithm to solve the sub-problem and further prove its convergence. The computational complexity of the proposed algorithm is also analyzed using the real floating point operation method~\cite{BookGolub1996}, showing a reasonable complexity. Finally, numerical results validate the effectiveness of the developed algorithm and show that our algorithm is able to achieve simultaneously both the maximum sum rate and the maximum energy efficiency at the low transmit power region, while at the middle-high transmit power region high sum rate does not necessarily brings high energy efficiency.

This rest of this paper  is organized as follows. The system model is described in Section \uppercase\expandafter{\romannumeral2}. In Section \uppercase\expandafter{\romannumeral3}, an energy efficient beamforming algorithm is proposed for the multicell MU-MISO downlink system subject to per-BS power constraints. The computational complexity and the parallel implementation of the proposed algorithm are analyzed in \uppercase\expandafter{\romannumeral4}. The simulation results are shown in Section \uppercase\expandafter{\romannumeral5} and conclusions are finally given in Section \uppercase\expandafter{\romannumeral6}.

The following notations are used throughout this paper. Bold lowercase and uppercase letters
represent column vectors and matrices, respectively. The superscript \textsuperscript{T},
\textsuperscript{H}, \textsuperscript{*} and \textsuperscript{\dag} represent the transpose operator, conjugate transpose
operator, conjugate operator and the Moore Penrose pseudo-inverse of matrix, respectively. $\bm{A}_{m,n}$ represents the $\left(m{\rm th},n{\rm th}\right)$ element of matrix $\bm{A}$. $\|\cdot\|$ denotes the $\ell_{2}$ norm.

\section*{\sc \uppercase\expandafter{\romannumeral2}. System Model}

As illustrated in Fig.~\ref{SystemModel}, consider a $K$-cell MU-MISO downlink system where BS-$j$ is equipped with $M_{j}$ transmit antennas and serves $N_{j}$ single-antenna users in cell $j$, $j=1,\cdots,K$. We denote the $k$-th user in cell $j$ as User-$\left(j,k\right)$ and BS in cell $m$ as BS-$m$. Then, the received signal of the User-$\left(j,k\right)$ is denoted as
\begin{equation}\label{EngergyEfficiency1}
y_{j,k}=\sum_{m=1}^{K}\bm{h}_{m,j,k}^{H}\sum_{n=1}^{N_{m}}\bm{w}_{m,n}x_{m,n}+z_{j,k}
\end{equation}
where $\bm{h}_{m,j,k}\in\mathbb{C}^{M_m}$ denotes the flat fading channel coefficient between BS-$m$ and User-$\left(j,k\right)$ which includes the large scale fading, the small scale fading and shadow fading, $\bm{w}_{j,k}$ denotes the beamforming vector for User-$\left(j,k\right)$, $x_{j,k}$ denotes the information signal intended for User-$\left(j,k\right)$ with $\mathbb{E}\left\{x_{j,k}\right\}=0$ and $\mathbb{E}\left\{\|x_{j,k}\|^{2}\right\}=1$, and $z_{j,k}$ is a zero-mean circularly symmetric complex Gaussian random noise with variance $\sigma_{j,k}^{2}$. We further assume that the signals for different users are independent from each other and the receiver noise. The instantaneous rate of User-$\left(j,k\right)$ is calculated as\footnote{The logarithm with $e$ as the base is used throughout this paper.}
\begin{equation}\label{EngergyEfficiency2}
R_{j,k}=\log\left(1+\frac{\|\bm{h}_{j,j,k}^{H}\bm{w}_{j,k}\|^{2}}
{\Upsilon_{j,k}+\sigma_{j,k}^{2}}\right)
\end{equation}
where $\Upsilon_{j,k}$ denotes the interference signal strength which includes the intra-cell inter-user interference signal strength and the inter-cell interference signal strength and is given by
\begin{equation}
\Upsilon_{j,k}=\sum\limits_{n=1, n\neq k}^{N_{j}}\|\bm{h}_{j,j,k}^{H}\bm{w}_{j,n}\|^{2}+
\sum\limits_{m=1, m\neq j}^{K}\sum_{n=1}^{N_{m}}\|\bm{h}_{m,j,k}^{H}\bm{w}_{m,n}\|^{2}.
\end{equation}
\begin{figure}[ht]
\centering
\captionstyle{flushleft}
\onelinecaptionstrue
\includegraphics[width=0.8\columnwidth,keepaspectratio]{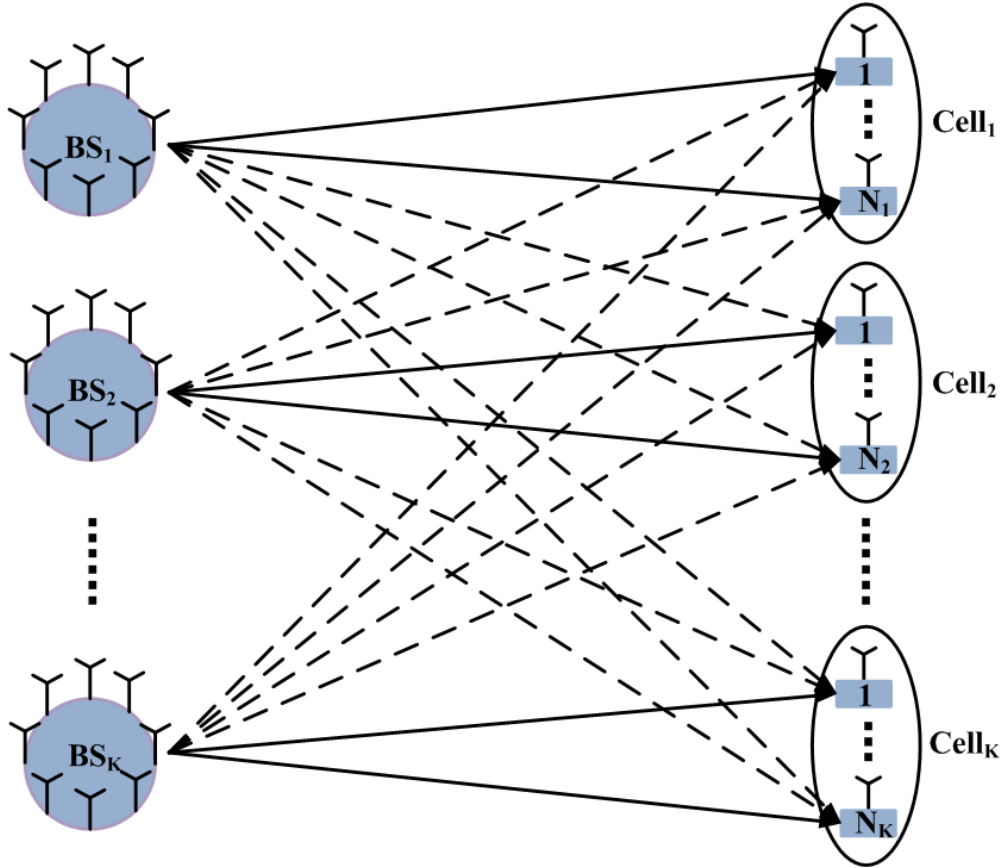}\\
\caption{The diagram of coordinated multicell beamforming system.}
\label{SystemModel}
\end{figure}

For notational convenience, let $\bm{W}_{j}=\left\{\bm{w}_{j,1},\cdots,\bm{w}_{j,N_{j}}\right\}$ denote the multiuser precoder set of BS-$j$ and let $\bm{W}=\left\{\bm{W}_{j},\cdots,\bm{W}_{K}\right\}$ denote the collection of all the precoders. The energy efficiency of interest is defined as the ratio of the weighted sum rate to the total power consumption, given by
\begin{equation}\label{EngergyEfficiency3}
f\left(\bm{W}\right)=\frac{f_{1}\left(\bm{W}\right)}{f_{2}\left(\bm{W}\right)}
=\frac{\sum\limits_{j,k}\alpha_{j,k}R_{j,k}}{\xi\sum\limits_{j,k}\|\bm{w}_{j,k}\|^{2}
+\sum\limits_{j}\left(M_{j}P_{c}+P_{0}\right)}
\end{equation}
where the weight $\alpha_{j,k}$ is used to represent the priority of User-$\left(j,k\right)$ in the system, $\xi\geq 1$ is a constant which accounts for the inefficiency of the power amplifier, $P_{c}$ is the constant circuit power consumption per antenna which includes the power dissipations in the transmit filter,  mixer, frequency synthesizer, and digital-to-analog converter, and $P_{0}$ is the basic power consumed at the BS which is independent of the number of transmit antennas~\cite{TWCNg201209,TWCNg201210}. In order to obtain a tradeoff between the sum rate and the total power consumption, the energy efficiency maximization performance criterion is adopted, given by
\begin{equation}\label{EngergyEfficiency5}
\max_{\bm{W}}~f\left(\bm{W}\right)
~s.t.~ \sum_{k=1}^{N_{j}}\|\bm{w}_{j,k}\|^{2}\leq P_{j},\forall j
\end{equation}
where $P_{j}$ is the transmit power constraint of BS-$j$. As a comparison, the traditional weighted sum rate optimization problem is usually defined as
\begin{equation}\label{MaxSumRate}
\max_{\bm{W}}~f_{1}\left(\bm{W}\right)
~s.t.~ \sum_{k=1}^{N_{j}}\|\bm{w}_{j,k}\|^{2}\leq P_{j},\forall j.
\end{equation}
Different from (\ref{EngergyEfficiency5}), in (\ref{MaxSumRate}), it is only concerned with how to maximize the sum rate, without taking the power consumption into account. Note that the coupling of optimization variables leads to that the problem (\ref{EngergyEfficiency5}) and (\ref{MaxSumRate}) become  non-convex and thus are difficult to solve directly. Furthermore, the fractional form of the objective function in (\ref{EngergyEfficiency5}) results in that common optimization approaches are not applicable now. Considering the transmit power constraints, the following inequalities are easily obtained.
\begin{subnumcases}{\label{EngergyEfficiency6}}
0\leq f_{1}\left(\bm{W}\right)\leq R_{max}\\
\sum\limits_{j}\left(M_{j}P_{c}+P_{0}\right)\leq f_{2}\left(\bm{W}\right)\\
f_{2}\left(\bm{W}\right)\leq \sum\limits_{j}\left(P_{j}+M_{j}P_{c}+P_{0}\right)
\end{subnumcases}
where $R_{max}=\sum\limits_{j,k}\log\left(1+\frac{P_{j}\|\bm{h}_{j,j,k}\|^{2}}{\sigma_{j,k}^{2}}\right)$ is the maximum rate achieved with the maximum transmit power and without considering the inter-cell interference and the intra-cell interference.

\section*{\sc \uppercase\expandafter{\romannumeral3}. Energy Efficient Beamforming Algorithm Design}

In this section, we will design a two-layer optimization scheme to solve the non-convex problem (\ref{EngergyEfficiency5}). By exploiting the relationship between the fractional and the parametric programming problems~\cite{JstorJagan1966,JstorDink1967,MathCrouzeix1991}, the original fractional problem is first transformed into an equivalent non-fractional problem. The equivalent problem is further cast into a tractable form using the relationship between the user rate and the MSE of the optimal receiver. Based on this, an optimization algorithm is finally developed to reach the solution to (\ref{EngergyEfficiency5}).

\subsection*{A. Equivalent  Optimization Problem}

It is easy to see that the optimization problem (\ref{EngergyEfficiency5}) belongs to a classical fractional programming problem. To solve it, as revealed in~\cite{JstorJagan1966,JstorDink1967,MathCrouzeix1991}, a common approach is to transform it into a linear programming problem by adopting nonlinear variable transformation. Following this idea, here we exploit the relationship between the fractional programming and the parametric programming problems to reformulate problem (\ref{EngergyEfficiency5}) into the following univariate equation
\begin{equation}\label{EngergyEfficiency7}
F\left(\eta\right)=0
\end{equation}
where the function $F:\mathbb{R}\longrightarrow\mathbb{R}$ is defined by
\begin{equation}\label{EngergyEfficiency8}
F\left(\eta\right)=\max_{\bm{W}\in \mathbb{D}}\left\{f_{1}\left(\bm{W}\right)-\eta f_{2}\left(\bm{W}\right)\right\}
\end{equation}
with $\mathbb{D}=\left\{\bm{W}\Big| \sum\limits_{k}^{N_{j}}\|\bm{w}_{j,k}\|^{2}\leq P_{j}, \forall j\right\}$. To clarify its equivalence to the primal problem, it is interesting to first note that the univariate function $F\left(\eta\right)$ has some especially pleasant properties summarized in the following theorem, which is similar to the results obtained in~\cite{JstorJagan1966,JstorDink1967}.
\begin{theorem}\label{EETheorem1}
Let $F:\mathbb{R}\longrightarrow\mathbb{R}$ be defined by (\ref{EngergyEfficiency8}). Then, the following statements hold.
\begin{itemize}
\item[(a)] $F$ is convex over $\mathbb{R}$.
\item[(b)] $F$ is continuous at any $\eta\in\mathbb{R}$.
\item[(c)] $F$ is strictly decreasing.
\item[(d)] $F\left(\eta\right)=0$ has a unique solution.
\end{itemize}
\end{theorem}
\begin{proof}
Let real numbers $\eta_{1}$ and $\eta_{2}$ be arbitrarily chosen so that $\eta_{1}\neq\eta_{2}$, and the corresponding optimal solutions of problem (\ref{EngergyEfficiency8}) are $\bm{W}^{1}$ and $\bm{W}^{2}$, $\forall i$, respectively.

(a) For any $0<\beta< 1$, let $\eta_{3}=\beta\eta_{1}+\left(1-\beta\right)\eta_{2}$ and the corresponding optimal solutions of problem (\ref{EngergyEfficiency7}) are $\bm{W}^{3}$, $\forall i$, we have
\begin{equation}\label{EngergyEfficiency9}
\begin{split}
&\beta F\left(\eta_{1}\right)+\left(1-\beta\right)F\left(\eta_{2}\right)\\
=&\beta \left(f_{1}\left(\bm{W}^{1}\right)-\eta_{1} f_{2}\left(\bm{W}^{1}\right)\right)\\
&+\left(1-\beta\right)\left(f_{1}\left(\bm{W}^{2}\right)-\eta_{2} f_{2}\left(\bm{W}^{2}\right)\right)\\
\geq&\beta \left(f_{1}\left(\bm{W}^{3}\right)-\eta_{1} f_{2}\left(\bm{W}^{3}\right)\right)\\
&+\left(1-\beta\right)\left(f_{1}\left(\bm{W}^{3}\right)-\eta_{2} f_{2}\left(\bm{W}^{3}\right)\right)\\
=&f_{1}\left(\bm{W}^{3}\right)- \left(\beta\eta_{1}+\left(1-\beta\right)\eta_{2}\right) f_{2}\left(\bm{W}^{3}\right)\\
=&F\left(\eta_{3}\right)
\end{split}
\end{equation}
where the inequality follows from the definition of $F\left(\eta\right)$. Based on the above results we see that $F\left(\eta\right)$ is a convex function.

(b) Since $F\left(\eta\right)$ is a convex mapping from $\mathbb{R}$ to $\mathbb{R}$, we can easily see the continuity.

(c) Similarly, the strict decreasing of $F\left(\eta\right)$ can be easily proven.

(d) According to (\ref{EngergyEfficiency6}), we have $F\left(\eta_{1}\right)\leq R_{max}-\eta_{1}KP_{c}$ and $F\left(\eta_{1}\right)\geq -K\left(\eta_{1}P_{c}+P_{j}\right)$. From these two equations, we can see that $\lim\limits_{\eta_{1}\to +\infty}F\left(\eta_{1}\right)=-\infty$ and $\lim\limits_{\eta_{1}\to -\infty}F\left(\eta_{1}\right)=+\infty$. Combining (b) and (c) yields the unique solvability of $F\left(\eta\right)=0$.
\end{proof}
The above theorem means that $F\left(\eta\right)$ is monotonically decreasing and the equation $F\left(\eta\right)=0$ has a unique solution. With these results, the equivalence between problem~(\ref{EngergyEfficiency5}) and (\ref{EngergyEfficiency7}) is given as the following proposition.
\begin{proposition}\label{EEProposition}
The following two statements are equivalent:
\begin{itemize}
\item[(a)] $\max\limits_{\bm{W}\in \mathbb{D}}f\left(\bm{W}\right)=\max\limits_{\bm{W}\in \mathbb{D}}\frac{f_{1}\left(\bm{W}\right)}{f_{2}\left(\bm{W}\right)}=\eta$
\item[(b)] $F\left(\eta\right)=\max\limits_{\bm{W}\in \mathbb{D}}\left\{f_{1}\left(\bm{W}\right)-\eta f_{2}\left(\bm{W}\right)\right\}=0$
\end{itemize}
\end{proposition}
\begin{proof}
First proving (a)$\Longrightarrow$(b). Let $\bm{W}^{opt}$ be the solution of the problem (\ref{EngergyEfficiency5}), for $\forall \bm{W}\in \mathbb{D}$, we have
\begin{equation}\label{EngergyEfficiency27}
\eta=f\left(\bm{W}^{opt}\right)=\frac{f_{1}\left(\bm{W}^{opt}\right)}{f_{2}\left(\bm{W}^{Opt}\right)}
\geq\frac{f_{1}\left(\bm{W}\right)}{f_{2}\left(\bm{W}\right)}
\end{equation}
According to (\ref{EngergyEfficiency6}), we easily know that $f_{2}\left(\bm{W}\right)> 0, \forall \bm{W}\in \mathbb{D}$ and then have the following equations
\begin{equation}\label{EngergyEfficiency28}
\begin{split}
&f_{1}\left(\bm{W}^{opt}\right)-\eta f_{2}\left(\bm{W}^{opt}\right)=0\\
&f_{1}\left(\bm{W}\right)-\eta f_{2}\left(\bm{W}\right)\leq 0
\end{split}
\end{equation}
Based on the above results, we can easily know that $\max\limits_{\bm{W}\in \mathbb{D}}\left\{f_{1}\left(\bm{W}\right)-\eta f_{2}\left(\bm{W}\right)\right\}=0$ and the maximum value is obtained at $\bm{W}^{opt}$.

Next proving (b)$\Longrightarrow$(a). Let $\bm{W}^{opt}$ be the solution of the problem (\ref{EngergyEfficiency8}), for $\forall \bm{W}\in \mathbb{D}$, then we have
\begin{equation}\label{EngergyEfficiency29}
\begin{split}
0=F\left(\eta\right)&=f_{1}\left(\bm{W}^{opt}\right)-\eta f_{2}\left(\bm{W}^{opt}\right)\\
&\geq f_{1}\left(\bm{W}\right)-\eta f_{2}\left(\bm{W}\right)
\end{split}
\end{equation}
From~(\ref{EngergyEfficiency6}) and~(\ref{EngergyEfficiency29}), we easily know that $f_{2}\left(\bm{W}\right)> 0, \forall \bm{W}\in \mathbb{D}$ and then have the following relations
\begin{equation}\label{EngergyEfficiency30}
\begin{split}
&\frac{f_{1}\left(\bm{W}^{opt}\right)}{ f_{2}\left(\bm{W}^{opt}\right)}=\eta\\
&\frac{f_{1}\left(\bm{W}\right)}{f_{2}\left(\bm{W}\right)}\leq \eta
\end{split}
\end{equation}
Based on that, we can easily know that $F\left(\eta\right)=\max\limits_{\bm{W}\in \mathbb{D}}\frac{f_{1}\left(\bm{W}\right)}{f_{2}\left(\bm{W}\right)}=\eta$ and the maximum value is obtained at $\bm{W}^{opt}$.
\end{proof}

The above proposition means that the univariate equation $F\left(\eta\right)=0$ is essentially equivalent to the primal fractional programming problem (\ref{EngergyEfficiency5}). In other words, if we can find a parameter $\eta$ such that the optimal value of problem (\ref{EngergyEfficiency8}) is zero, then the optimal solution of problem (\ref{EngergyEfficiency8}) is also the optimal solution of problem~(\ref{EngergyEfficiency5}). Henceforth the parameter $\eta$ is named as the energy efficiency factor of our considered communication systems. According to~(\ref{EngergyEfficiency6}), we have $0\leq\eta \leq \frac{R_{max}}{\sum\limits_{j}\left(M_{j}P_{c}+P_{0}\right)}$\footnote{It is seen that the value of circuit power, i.e., $P_c$ and $P_0$ affects the range of $\eta$.}. Combining Theorem~\ref{EETheorem1} with proposition~\ref{EEProposition}, problem~(\ref{EngergyEfficiency5}) can be solved by sequentially looking for the optimal univariate parameter $\eta$. Recalling the properties summarized in Theorem 1, it is easy to understand that one dimension search method is efficient to find the solution of $F\left(\eta\right)=0$, here we would like to employ the bi-section method~\cite{BookBoyd2004}. The corresponding iterative algorithm is summarized as Algorithm \ref{OuterLayer}.

\begin{algorithm}
\caption{Outer Layer Solution}\label{OuterLayer}
\begin{algorithmic}[1]
\STATE Initialize $\eta_{min}=0$, and $\eta_{max}=\frac{R_{max}}{\sum\limits_{j}\left(M_{j}P_{c}+P_{0}\right)}$.\label{li:Initial}
\STATE Let $\eta=\frac{\eta_{min}+\eta_{max}}{2}$, then solve problem (\ref{EngergyEfficiency8}) for given $\eta$, obtain the optimal solution $\{\bm{W}^{opt}\}$ and $F\left(\eta\right)$.\label{li:Sub-Problem}
\STATE If $F\left(\eta\right)\leq 0$, let $\eta_{max}=\eta$. Otherwise, let $\eta_{min}=\eta$.
\STATE if $\left|\eta_{max}-\eta_{min}\right|\leq\varepsilon$, where $\varepsilon$ is a predefined threshold, then stop. Otherwise, return to step \ref{li:Sub-Problem}.\label{li:Output}
\end{algorithmic}
\end{algorithm}

\subsection*{B. Solution of Sub-Problem}

It is easily seen that the key step in Algorithm~\ref{OuterLayer} lies in solving the sub-problem (\ref{EngergyEfficiency8}) to achieve the beamformers. Without loss of generality, we assume that the value of $\eta$ is greater than zero, and rewrite the sub-problem (\ref{EngergyEfficiency8}) into the following equivalent form for a given $\eta$
\begin{equation}\label{EngergyEfficiency13}
\begin{split}
&\max_{\bm{W}}~G\left(\bm{W}\right)=\sum\limits_{j,k}\left(\alpha_{j,k}R_{j,k}-\eta\xi\|\bm{w}_{j,k}\|^{2}\right)\\
~s.t.~& \sum_{k=1}^{N_{j}}\|\bm{w}_{j,k}\|^{2}\leq P_{j},\forall j.
\end{split}
\end{equation}
Although the problem (\ref{EngergyEfficiency13}) has a non-fractional form, it is still non-convex and its optimization variables are coupled. Next we further reformulate it into a tractable form by exploiting the relationship between the achievable rate and the MSE of the optimal receiver. We consider a linear receiver filter where the estimated signal is calculated as $\tilde{x}_{j,k}=\mu_{j,k}^{*}y_{i,k}$, with $\mu_{j,k}$ denoting the receiver filter at User-$\left(j,k\right)$. Then, the MSE $e_{j,k}$ for User-$\left(j,k\right)$ is calculated as
\begin{equation}\label{EngergyEfficiency14}
\begin{split}
e_{j,k}=&\mathbb{E}\left\{\left(\tilde{x}_{j,k}-x_{j,k}\right)\left(\tilde{x}_{j,k}-x_{j,k}\right)^{*}\right\}\\
=&\left|\mu_{j,k}\right|^{2}\left(\sum_{m,n}\left|\bm{h}_{m,j,k}^{H}
\bm{w}_{m,n}\right|^{2}+\sigma_{j,k}^{2}\right)\\
&-\mu_{j,k}\bm{w}_{j,k}^{H}\bm{h}_{j,j,k}-\mu_{j,k}^{*}\bm{h}_{j,j,k}^{H}\bm{w}_{j,k}+1.
\end{split}
\end{equation}
Let $\hat{e}_{j,k}$ denote the MSE achieved by the optimal receiver filter, then the user rate $R_{j,k}$ can be expressed as $R_{j,k}=\log\left(\frac{1}{\hat{e}_{j,k}}\right)$~\cite{TWCChristensen2008}. By introducing two sets of auxiliary variables $\bm{S}=\left\{\bm{s}_{1},\cdots,\bm{s}_{K}\right\}$, $\bm{s}_{j}=\left\{s_{j,1},\cdots,s_{j,N_{j}}\right\}$, $\forall j$, problem (\ref{EngergyEfficiency13}) is reformulated as follows by using Lemma~2 in~\cite{ICCJose2011}
\begin{equation}\label{EngergyEfficiency17}
\begin{split}
&\max_{\bm{W},\bm{U},\bm{S}}~\mathcal{H}\left(\bm{W},\bm{U},\bm{S}\right)\\
~s.t.~&\sum_{k=1}^{N_{j}}\|\bm{w}_{j,k}\|^{2}\leq P_{j},\forall j
\end{split}
\end{equation}
where the function $\mathcal{H}\left(\bm{W},\bm{U},\bm{S}\right)$ is defined as
\begin{equation}\label{EngergyEfficiency31}
\mathcal{H}\left(\bm{W},\bm{U},\bm{S}\right)=
\sum\limits_{j,k}\left(-\alpha_{j,k}e_{j,k}s_{j,k}+\alpha_{j,k}\log s_{j,k}
+\alpha_{j,k}-\eta\xi\|\bm{w}_{j,k}\|^{2}\right)
\end{equation}
and $\bm{U}=\left\{\bm{\mu}_{1},\cdots,\bm{\mu}_{K}\right\}$ with $\bm{\mu}_{j}=\left\{\mu_{j,1},\cdots,\mu_{j,N_{j}}\right\}$ denotes the receiver filters. One can see that compared with the primal problem,  the objective function of the equivalent problem~(\ref{EngergyEfficiency17}) has a more tractable form while introduces a few extra optimization variables. Combining~(\ref{EngergyEfficiency14}) and~(\ref{EngergyEfficiency31}), it is easily seen that the cost function $\mathcal{H}\left(\bm{W},\bm{U},\bm{S}\right)$ is convex in each of the optimization variables\footnote{Note that on the right hand side of (\ref{EngergyEfficiency31}) $e_{j,k}$ is a function of $\bm{U}$, as shown in (\ref{EngergyEfficiency14})} $\bm{W}$, $\bm{U}$, $\bm{S}$. In what follows, we propose to use the block coordinate descent method to solve problem~(\ref{EngergyEfficiency17}). Specifically, we maximize the cost function $\mathcal{H}\left(\bm{W},\bm{U},\bm{S}\right)$ by sequentially fixing two of the three variables $\bm{W}$, $\bm{U}$, $\bm{S}$ and updating the third. The optimal receiver filters $\bm{U}$ and the optimal auxiliary variables $\bm{S}$ are given by the following theorem for a given $\bm{W}$, respectively.
\begin{theorem}\label{EETheorem2}
For any given $\bm{W}$, the optimal receiver filters of the sub-problem (\ref{EngergyEfficiency17})
are given by
\begin{equation}\label{EngergyEfficiency18}
u_{j,k}^{opt}=\frac{\bm{h}_{j,j,k}^{H}\bm{w}_{j,k}}{\sum\limits_{m,n}\left|\bm{h}_{m,j,k}^{H}\bm{w}_{m,n}\right|^{2}+\sigma_{j,k}^{2}}, \forall j,k
\end{equation}
Furthermore, the optimal $s_{j,k}$ is given by
\begin{equation}\label{EngergyEfficiency19}
s_{j,k}^{opt}=\frac{1}{\hat{e}_{j,k}}, \forall j,k
\end{equation}
where $\hat{e}_{j,k}$ is given by
\begin{equation}\label{EngergyEfficiency16}
\hat{e}_{j,k}=1-\frac{\left|\bm{h}_{j,j,k}^{H}\bm{w}_{j,k}\right|^{2}}
{\sum\limits_{m,n}\left|\bm{h}_{m,j,k}^{H}\bm{w}_{m,n}\right|^{2}+\sigma_{j,k}^{2}}.
\end{equation}
\end{theorem}
\begin{proof}
For any given $\bm{W}$, first substitute the expression of $e_{j,k}$ into $\mathcal{H}\left(\bm{W},\bm{U},\bm{S}\right)$, then check the first optimality condition to find the optimal receiver filter $u_{j,k}$ and the optimal auxiliary variable $s_{j,k}$, i.e., $\frac{\partial\mathcal{H}\left(\bm{W},\bm{U},\bm{S}\right)}{\partial u_{j,k}^{*}}=0$ and $\frac{\partial\mathcal{H}\left(\bm{W},\bm{U},\bm{S}\right)}{\partial s_{j,k}}=0$, respectively. It follows that the optimal receiver filter is given as (\ref{EngergyEfficiency18}). By replacing the optimal receiver given in (\ref{EngergyEfficiency18}) into $s_{j,k}^{opt}=\frac{1}{e_{j,k}}$,  the result in (\ref{EngergyEfficiency16}) is obtained.
\end{proof}

Once the values of $\bm{U}$ and $\bm{S}$ are given, the optimization of $\bm{W}$ is decoupled among the BSs by substituting the expression of $e_{j,k}$ into $\mathcal{H}\left(\bm{W},\bm{U},\bm{S}\right)$, leading to the following parallel optimization problems given by
\begin{equation}\label{EngergyEfficiency20}
\begin{split}
\max_{\bm{W}_{j}}&-\sum_{k=1}^{N_{j}}\sum_{m,n}\alpha_{m,n}s_{m,n}\left|\mu_{m,n}\right|^{2}\left|\bm{h}_{j,m,n}^{H}
\bm{w}_{j,k}\right|^{2}\\
&+\sum_{k=1}^{N_{j}}\left(\alpha_{j,k}s_{j,k}\mu_{j,k}\bm{w}_{j,k}^{H}\bm{h}_{j,j,k}
-\eta\xi\|\bm{w}_{j,k}\|^{2}\right)\\
&+\sum_{k=1}^{N_{j}}\alpha_{j,k}s_{j,k}\mu_{j,k}^{*}\bm{h}_{j,j,k}^{H}\bm{w}_{j,k}\\
s.t.~&~ \sum_{k=1}^{N_{j}}\|\bm{w}_{j,k}\|^{2}\leq P_{j}.
\end{split}
\end{equation}
It is easily known that the above problem is a convex quadratic optimization problem which can be solved using standard approaches such as classical interior point method or second order conic programming (SOCP)~\cite{BookBoyd2004}. More importantly, we reveal that its solution has a closed-form expression using the Lagrange multiplier method.  We proceed by introducing a Lagrange multiplier $\lambda_{j}$ associated with the power constraint of BS-$j$, the corresponding Lagrange function can be written as
\begin{equation}\label{EngergyEfficiency21}
\begin{split}
\mathcal{L}\left(\bm{W}_{j},\lambda_{j}\right)=&-\sum_{k=1}^{N_{j}}\left(\sum_{m,n}
\alpha_{m,n}s_{m,n}\left|\mu_{m,n}\right|^{2}\left|\bm{h}_{j,m,n}^{H}\bm{w}_{j,k}\right|^{2}
-\alpha_{j,k}s_{j,k}\mu_{j,k}\bm{w}_{j,k}^{H}\bm{h}_{j,j,k}\right)\\
&+\sum_{k=1}^{N_{j}}\left(\alpha_{j,k}s_{j,k}\mu_{j,k}^{*}\bm{h}_{j,j,k}^{H}\bm{w}_{j,k}
-\eta\xi\|\bm{w}_{j,k}\|^{2}\right)-\lambda_{j}\left(\sum_{k=1}^{N_{j}}\|\bm{w}_{j,k}\|^{2}-P_{j}\right)
\end{split}
\end{equation}
The first-order optimality condition of $\mathcal{L}\left(\bm{W}_{j},\lambda_{j}\right)$ with respect to $\bm{w}_{j,k}^{H}$ yields
\begin{equation}\label{EngergyEfficiency22}
\bm{w}_{j,k}=\alpha_{j,k}s_{j,k}u_{j,k}\left(\bm{A}_{j}+\lambda_{j}\bm{I}\right)^{\dag}\bm{h}_{j,j,k}
\end{equation}
where $\bm{A}_{j}=\sum\limits_{m=1}^{K}\sum\limits_{n=1}^{N_{m}}\alpha_{m,n}s_{m,n}\left|u_{m,n}\right|^{2}
\bm{h}_{j,m,n}\bm{h}_{j,m,n}^{H}+\eta\xi\bm{I}$. Different from traditional beamforming design~\cite{TSPShi2011}, it is seen that the energy efficiency factor $\eta$ is included in $\bm{A}_{j}$, and $\lambda_{j}\geq 0$ should be chosen such that the complementary slackness condition of the power constraint is satisfied. For notational simplicity, we introduce a parametric representation for $\bm{w}_{j,k}$ and let $\bm{w}_{j,k}\left(\lambda_{j}\right)$ denote the right-hand side of (\ref{EngergyEfficiency22}) with parameter $\lambda_{j}$. Since $\eta> 0$ and $\xi> 0$, $\bm{A}_{j}$ is a positive-definite matrix, and $\bm{A}_{j}+\lambda_{j}\bm{I}$ is also a positive-definite matrix. Without loss of generality, we denote the eigendecomposition of matrix $\bm{A}_{i}$ as $\bm{\Phi}_{j}\bm{\Lambda}_{j}\bm{\Phi}_{j}^{H}$ and have
\begin{equation}\label{EngergyEfficiency24}
\left\|\bm{w}_{j,k}\left(\lambda_{i}\right)\right\|^{2}=
Tr\left(\left(\bm{\Lambda}_{j}+\lambda_{j}\bm{I}\right)^{-2}\bm{\Phi}_{j}^{H}\bm{\Psi}_{j,k}\bm{\Phi}_{j}\right)
\end{equation}
where $\bm{\Psi}_{j,k}=\left|\alpha_{j,k}s_{j,k}u_{j,k}\right|^{2}\bm{h}_{j,j,k}\bm{h}_{j,j,k}^{H}$. Let $\bm{\Psi}_{j}=\bm{\Phi}_{j}^{H}\left(\sum_{k=1}^{N_{j}}\bm{\Psi}_{j,k}\right)\bm{\Phi}_{j}$, then we have
\begin{equation}\label{EngergyEfficiency25}
\varphi\left(\lambda_{j}\right)=\sum_{k=1}^{N_{j}}\left\|\bm{w}_{j,k}\left(\lambda_{j}\right)\right\|^{2}
=\sum_{m=1}^{M_{j}}\frac{\left[\bm{\Psi}_{j}\right]_{m,m}}
{\left(\left[\bm{\Lambda}_{j}\right]_{m,m}+\lambda_{j}\right)^{2}}.
\end{equation}
It is easily observed that the function $\varphi\left(\lambda_{j}\right)$ is monotonically decreasing in $\lambda_{j}$ for $\lambda_{j}\geq 0$. If $\sum\limits_{k=1}^{N_{j}}\left\|\bm{w}_{j,k}\left(0\right)\right\|^{2}\leq P_{j}$, then $\bm{w}_{i,k}^{opt}=\bm{w}_{j,k}\left(0\right)$, $\forall k$, otherwise, we must have
\begin{equation}\label{EngergyEfficiency23}
\varphi\left(\lambda_{j}\right)=\sum_{k=1}^{N_{j}}\left\|\bm{w}_{j,k}\left(\lambda_{j}\right)\right\|^{2}=P_{j}.
\end{equation}
According to the monotonic property of the function $\varphi\left(\lambda_{j}\right)$ with respect to $\lambda_{j}$, the equation (\ref{EngergyEfficiency23}) can be solved by one dimension search method, such as the bi-section method~\cite{BookBoyd2004}. Once the optimal $\lambda_{j}^{opt}$ is obtained, the optimal beamformer $\bm{w}_{j,k}^{opt}$ can also be calculated by (\ref{EngergyEfficiency22}), i.e.,
\begin{equation}\label{EngergyEfficiency26}
\bm{w}_{j,k}^{opt}=\alpha_{j,k}s_{j,k}u_{j,k}\left(\bm{A}_{j}+\lambda_{j}^{opt}\bm{I}\right)^{\dag}\bm{h}_{j,j,k}.
\end{equation}
Based on the above analysis,  the alternating optimization strategy can now be used to reach the solution of sub-problem (\ref{EngergyEfficiency17}), summarized as the following Algorithm~\ref{SubProblemSolution}.
\begin{algorithm}[H]
\caption{Sub-Problem Solution}\label{SubProblemSolution}
\begin{algorithmic}[1]
\STATE Set $n=0$, initialize $\bm{w}_{j,k}^{(n)}$ such that $\sum\limits_{k=1}^{N_{j}}\left\|\bm{w}_{j,k}^{(n)}\right\|^{2}\leq P_{j}$, $u_{j,k}^{(n)}=0$, $s_{j,k}^{(n)}=0$, $\forall j,k$, and compute $G\left(\bm{W}^{(n)}\right)=0$. \label{Initial}
\STATE Let $n=n+1$, update $u_{j,k}$ with (\ref{EngergyEfficiency18}), and obtain $u_{j,k}^{(n)}$, $\forall j,k$.\label{UpdatedU}
\STATE Update $s_{j,k}$ with (\ref{EngergyEfficiency19}) and (\ref{EngergyEfficiency16}), and obtain $s_{j,k}^{(n)}$ and $\hat{e}_{j,k}^{(n)}$, $\forall j,k$.\label{UpdatedS}
\STATE Update $\bm{w}_{j,k}$ with (\ref{EngergyEfficiency26}), and obtain $\bm{w}_{j,k}^{(n)}$, $\forall j,k$. \label{UpdatedW}
\STATE If $\left|G\left(\bm{W}^{(n)}\right)-G\left(\bm{W}^{(n-1)}\right)\right|\leq\delta$, where $\delta$ is a predefined threshold, and stop the algorithm. Otherwise return to step \ref{UpdatedU}.\label{OutputResult}
\end{algorithmic}
\end{algorithm}
\begin{theorem}\label{EETheorem3}
Algorithm \ref{SubProblemSolution} is guaranteed to converge.
\end{theorem}
\begin{proof}
Since the updates at step \ref{UpdatedU}, step \ref{UpdatedS}, and step \ref{UpdatedW} all  maximize the target object $\mathcal{H}\left(\bm{W},\bm{U},\bm{S}\right)$ at each iteration, i.e., maximize the target object $G\left(\bm{W}\right)$, the iterations in Algorithm~\ref{SubProblemSolution} lead to monotone increase of the object function (\ref{EngergyEfficiency17}). Since the achievable rate region under the practical per-BS power constraints is
bounded, i.e., the objective function is bounded, this monotonicity guarantees the convergence of the algorithm~\cite{Bibby1974}.
\end{proof}

By combining Algorithm~\ref{OuterLayer} and Algorithm~\ref{SubProblemSolution}, the energy efficient optimization problem (\ref{EngergyEfficiency5}) can now be efficiently solved. More importantly, we also derive closed-form expressions for the optimal beamformers, the optimal receiver filters and the auxiliary variables, which provide some insight on the energy-efficient optimization.

\begin{remark}
\rm It is interesting to note that our developed energy efficient optimization algorithm can be easily extended to multicell multiuser MIMO downlink system by simply reformulating the MMSE expression, the corresponding MMSE receiver filters, the corresponding auxiliary variable expressions, and the transmit beamforming matrices. Also, note that in this paper we only aim to maximize the energy efficiency without considering the rate requirements of individual cells or individual users. In some scenarios, we may need to satisfy additional rate constraints while maximizing the energy efficiency. Furthermore, our developed ideas can be easily applied to this case by adding the additional rate constraints in~(\ref{EngergyEfficiency5}), which we would like to leave as future work.
\end{remark}

\section*{\sc \uppercase\expandafter{\romannumeral4}.Computational Complexity  and Implementation Analysis}

In this section, we discuss the implementation issue of our proposed algorithm for  problem (\ref{EngergyEfficiency5}).
The computational complexity of our algorithm is first analyzed using the real floating point operation method. Besides, it is also shown that our algorithm can be carried out in a decentralized or parallel manner.

\subsection*{A. Complexity Analysis}

In what follows, the computational complexity is measured by counting the number of flops defined as real floating point operation~\cite{BookGolub1996}. That is, a real addition, multiplication, or division is counted as one flop, while a complex addition and multiplication have two flops and six flops, respectively. It is easy to see that the major computation of the proposed energy efficient beamforming algorithm lies in the execution of Algorithm 2. Let $N=\sum\limits_{j=1}^{K}N_{j}$ denote the total number of users in $K$ cells and $L=\sum\limits_{j=1}^{K}N_{j}M_{j}$, the computational complexity involved in Algorithm~\ref{SubProblemSolution} is counted as
\begin{itemize}
\item Updating the receiver filter at step \ref{UpdatedU} involves about $\phi_{1}=9NL$ flops.
\item Step \ref{UpdatedS} needs $\phi_{2}=8\left(N+2\right)L+3N$ flops.
\item In Step \ref{UpdatedW}, computing matrix  $\bm{A}_{j}$ requires $\left(N+1\right)M_{j}+8M_{j}^{2}+12N$ flops. As $\bm{A}_{j}$ is a positive definite matrix, Cholesky decomposition can be used to simplify the following inversion step. Thus, the flop count for $\left(\bm{A}_{j}+\lambda_{j}\right)^{-1}$ is about $\frac{8}{3}M_{j}^{3}+7M_{j}^{2}+\frac{7}{6}M_{j}$ and eigenvalue decomposition needs about $126M_{j}^{3}$ flops~\cite{TComZhang2008}. Since the bisection step generally takes a few iterations, here we ignore its effect in the complexity analysis. Thus, the flop count for step \ref{UpdatedW} is about $\phi_{3}=\sum\limits_{j=1}^{K}\left(129M_{j}^{3}+\left(15+8N_{j}\right)M_{j}^{2}+(N+2)M_{j}\right)+\left(12K+8\right)N$.
\end{itemize}
Therefore, finding the solution of the proposed method for problem~(\ref{EngergyEfficiency5}) takes about $\varrho_{1}\varrho_{2}\sum\limits_{i=1}^{3}\phi_{i}$ flops in total, where $\varrho_{1}$ and $\varrho_{2}$ denote the number of iterations for Algorithm \ref{OuterLayer} and Algorithm~\ref{SubProblemSolution}, respectively.

It is useful to compare the complexity of the proposed algorithm with some baseline algorithms such as the sum rate maximization beamforming algorithm proposed in~\cite{TSPShi2011}, which is based on iterative minimization of weighted MMSE (WMMSE). One can see that for the execution of one iteration, the WMMSE algorithm has the same order of computational complexity as our proposed Algorithm~\ref{SubProblemSolution}. That is, the execution of the WMMSE algorithm approximately takes $\varrho_{3}\sum\limits_{i=1}^{3}\phi_{i}$ flops, where $\varrho_{3}$ is the number of iterations needed to reach the convergence condition. Though in most cases the proposed algorithm may require a greater number of iterations than the WMMSE algorithm, due to the fact that in our algorithm the additional energy efficiency factor $\eta$ needs to be updated. However, owing to the property summarized in Theorem 1, the search for the optimal $\eta$ generally takes a few iterations, therefore it will not significantly increase the complexity.

\subsection*{B. Parallel Implementation}

From the steps of Algorithm \ref{SubProblemSolution}, we can find that the beamforming vectors can be optimized in parallel. In other words, if all channel state information (CSI) is collected at a central controller which has $K$ parallel processors that can exchange information with each other\footnote{The $K$ parallel processors can be the $K$ BSs if they are linked with capacity sufficient backhaul. In some practical scenarios there may not exist a central controller, in this case one of the BSs needs to be assigned as the controller and leads the algorithm implementation.}, the beamforming vectors can be simultaneously updated by $K$ different parallel processors. Thus, the developed algorithm can be implemented in a parallel compute fashion with a similar method used in~\cite{Bogale201201}. The detail steps are described as follows.
\begin{enumerate}
\item As the initialization, the central controller collects all CSI and share them among $K$ parallel processors. Let $n=0$, the central controller distributes an initial $\eta^{(n)}$ to all $K$ parallel processors.\label{ParallelStep01}
\item Each parallel processor $j$ optimizes its beamforming vector $\bm{W}_{j}^{(n+1)}$, the corresponding user's receiver filter $\bm{\mu}_{j}^{(n+1)}$, the auxiliary variable $\bm{s}_{j}^{(n+1)}$ according to Algorithm~\ref{SubProblemSolution} for the fixed $\eta^{(n)}$. \label{ParallelStep02}
\item Let $n=n+1$, the central controller updates the value of $\eta$ according to Algorithm~\ref{OuterLayer}, obtaining an updated $\eta^{(n)}$. \label{ParallelStep03}
\item The central controller determines whether to stop the iteration or not. If yes, then the central controller sends a stop command to each parallel processor. Otherwise, send $\eta^{(n)}$ to all parallel processors and return to Step 2.\label{ParallelStep04}
\end{enumerate}

It is seen from the above steps that the main system overhead of the proposed algorithm consists of the following parts. In the initialization, all CSI needs to be shared among $K$ parallel processors which causes signalling overhead, and the central controller shall distribute the energy efficiency factor $\eta$ to all parallel processors. In each iteration, as seen from Algorithm~\ref{SubProblemSolution}, updating the corresponding variables in Step~\ref{ParallelStep02} requires that all parallel processors send the updated parameters to the central controller and share them among $K$ parallel processors, including $\sum_{k=1}^{N_{j}}\left|\bm{h}_{j,m,n}^{H}\bm{w}_{j,k}\right|^{2}$, $s_{j,k}$, and $\left|\mu_{j,k}\right|^2$. Also, as seen from Algorithm~\ref{OuterLayer}, updating the energy efficiency factor $\eta^{(n)}$ in Step~\ref{ParallelStep03} requires that all parallel processors send the updated transmit power, i.e., $\sum_{k=1}^{N_{j}}\left\|\bm{w}_{j,k}\right\|^{2}$, to the central controller.  In summary, the execution of the algorithm requires a total of $\kappa_{1}\left(3\kappa_{2}\sum_{j=1}^{K}N_{j}+K+1\right)$ real numbers to be shared among all parallel processors and the central controller, where $\kappa_{1}$ and $\kappa_{2}$ denote the running number of Step~\ref{ParallelStep02} and Step~\ref{ParallelStep04} required by the convergence of the algorithm, respectively.

\section*{\sc \uppercase\expandafter{\romannumeral5}. Simulation Results}

In this section, we investigate the performance of the proposed multicell beamforming algorithm via numerical simulations. We consider a cooperative cluster of $K=3$ hexagonal adjacent cells where each BS-$j$ is equipped with $M_{j}$ transmit antennas and serves $N_{j}$ single antenna users in cell $j$. The cell radius is set to be $500$ m and each user has at least $400$ m distance from its serving BS. The channel vector $\bm{h}_{m,j,k}$ from BS $m$ to User-$\left(j,k\right)$ is generated based on the formulation $\bm{h}_{m,j,k}\triangleq\sqrt{\theta_{m,j,k}}\bm{h}_{m,j,k}^{w}$, where $\bm{h}_{m,j,k}^{w}$ denotes the small scale fading part and is assumed to be Gaussian distributed with zero mean and identity covariance matrix, and $\theta_{m,j,k}$ denotes the large scale fading factor which in decibels is given as $10\log_{10}(\theta_{m,j,k})=-38\log_{10}(d_{m,j,k})-34.5+\eta_{m,j,k}$, where $\eta_{m,j,k}$ represents the log-normal shadow fading with zero mean and standard deviation $8$ dB~\cite{3GPP}. The circuit power per antenna is $P_{c}=30$ dBm~\cite{ConKumar2011}, and the basic power consumed at the BS is $P_{0}=40$ dBm~\cite{ConArnold2010}. As for the power constraints, we assume that each BS has the same power constraint over the whole bandwidth. The noise figure is $9$ dB. The weighted factor $\alpha_{j,k}$ is set to unit for any $j$ and $k$. The inefficiency factor of power amplifier $\xi$ is set to unit. The convergence thresholds are given as $\delta=10^{-3}$ and $\varepsilon=10^{-5}$.

For comparison, the performance of the WMMSE algorithm which aims to only maximize the sum rate is simulated~\cite{TSPShi2011}. In addition, the performance of the power allocation algorithm which aims to maximize the energy efficiency by optimizing only the transmit power for fixed transmit beamformers, such as the maximum ratio transmission (MRT) beamforming or random beamforming, is investigated too, which is simulated by using jointly the fractional programming method and the convex approximation power allocation method in~\cite{TSPCod2007}.

Note that though the convergence of our proposed iterative scheme is proved, the global optimality cannot be guaranteed. Therefore it is important to examine the gap of the proposed solution from the optimum. Fig.~\ref{SubProblemAffections} shows the energy efficiency performance of the proposed scheme under a few random channel realizations for the configuration $P_{j}=46$ dBm, $\forall j$. The optimal energy efficiency is achieved by solving the sub-problem (\ref{EngergyEfficiency17}) for each $\eta$ with Algorithm~\ref{SubProblemSolution} over $10000$ random beamforming initializations and then choosing the best result. Numerical results corroborate that in most cases, our proposed energy efficiency optimization algorithm can always achieve over $96\%$ of the optimal performance, revealing that our solution achieves a near-optimal energy efficient performance. Fig.~\ref{ConvergenceTrajectory} shows the convergence behavior of Algorithm~\ref{SubProblemSolution} under a few random channel realizations for the configuration $P_{j}=46$ dBm, $\forall j$, using random beamforming initialization. The optimal sum rate is the best result among the rates achieved by Algorithm~\ref{SubProblemSolution} over $10000$ random beamforming initializations. Numerical results show that Algorithm~\ref{SubProblemSolution} always converges to a stable point in a limited number of iterations. Though it is observed that different initialization points may have slightly different effect on the performance and the convergence speed, our algorithm always achieves over $99\%$ of the optimal sum rate performance.

\begin{figure}[h]
\centering
\captionstyle{flushleft}
\onelinecaptionstrue
\includegraphics[width=0.8\columnwidth,keepaspectratio]{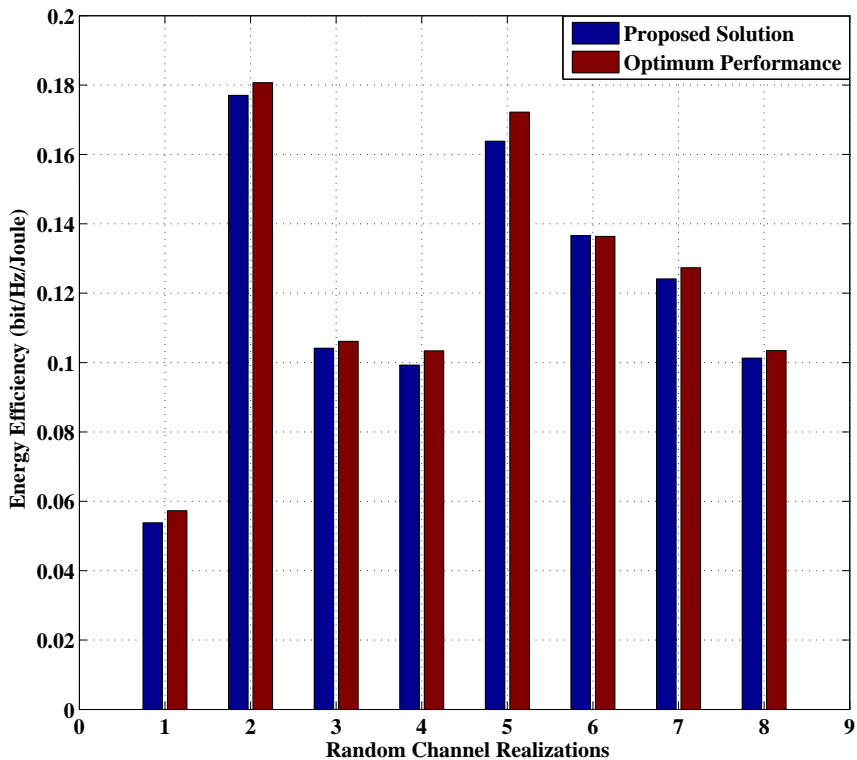}\\
\caption{The energy efficiency of the proposed solution in contrast to the optimum, $M_{j}=4$, $N_{j}=1$, $\forall j$.}
\label{SubProblemAffections}
\end{figure}

\begin{figure}[h]
\centering
\captionstyle{flushleft}
\onelinecaptionstrue
\includegraphics[width=0.8\columnwidth,keepaspectratio]{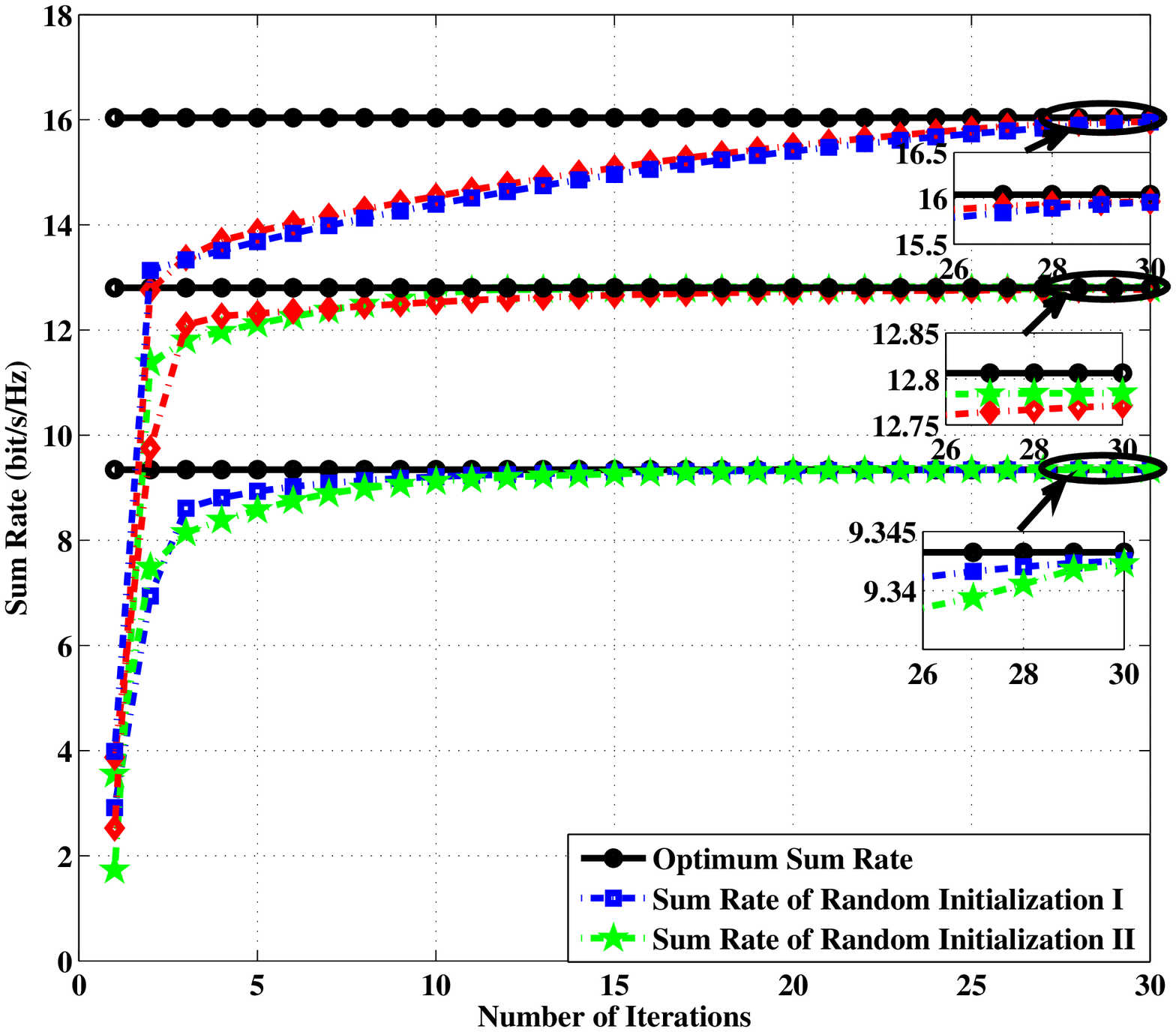}\\
\caption{Convergence Trajectory of Algorithm \ref{SubProblemSolution}, $M_{j}=4$, $N_{j}=1$, $\forall j$.}
\label{ConvergenceTrajectory}
\end{figure}

Fig.~\ref{MultiUsersE} shows the average energy efficiency of the proposed algorithm and the WMMSE algorithm with different user configurations over $10000$ random channel realizations. The results show that at the low transmit power region such as $26\sim34$ dBm, these two algorithms achieve almost the same energy efficiency, which suggests that at this region, transmitting with the maximum available power is the most energy efficient. It is also shown that the energy efficiency of the proposed algorithm obviously outperforms the WMMSE algorithm at the high transmit power region. This is because in the WMMSE algorithm the capacity gain cannot compensate for the negative impact of the maximum power consumption, resulting in a low energy efficiency. Numerical results also illustrate that the average energy efficiency increases as the number of served users grows, but the performance gain shrinks with the number of served users increases.

\begin{figure}[h]
\centering
\captionstyle{flushleft}
\onelinecaptionstrue
\includegraphics[width=0.8\columnwidth,keepaspectratio]{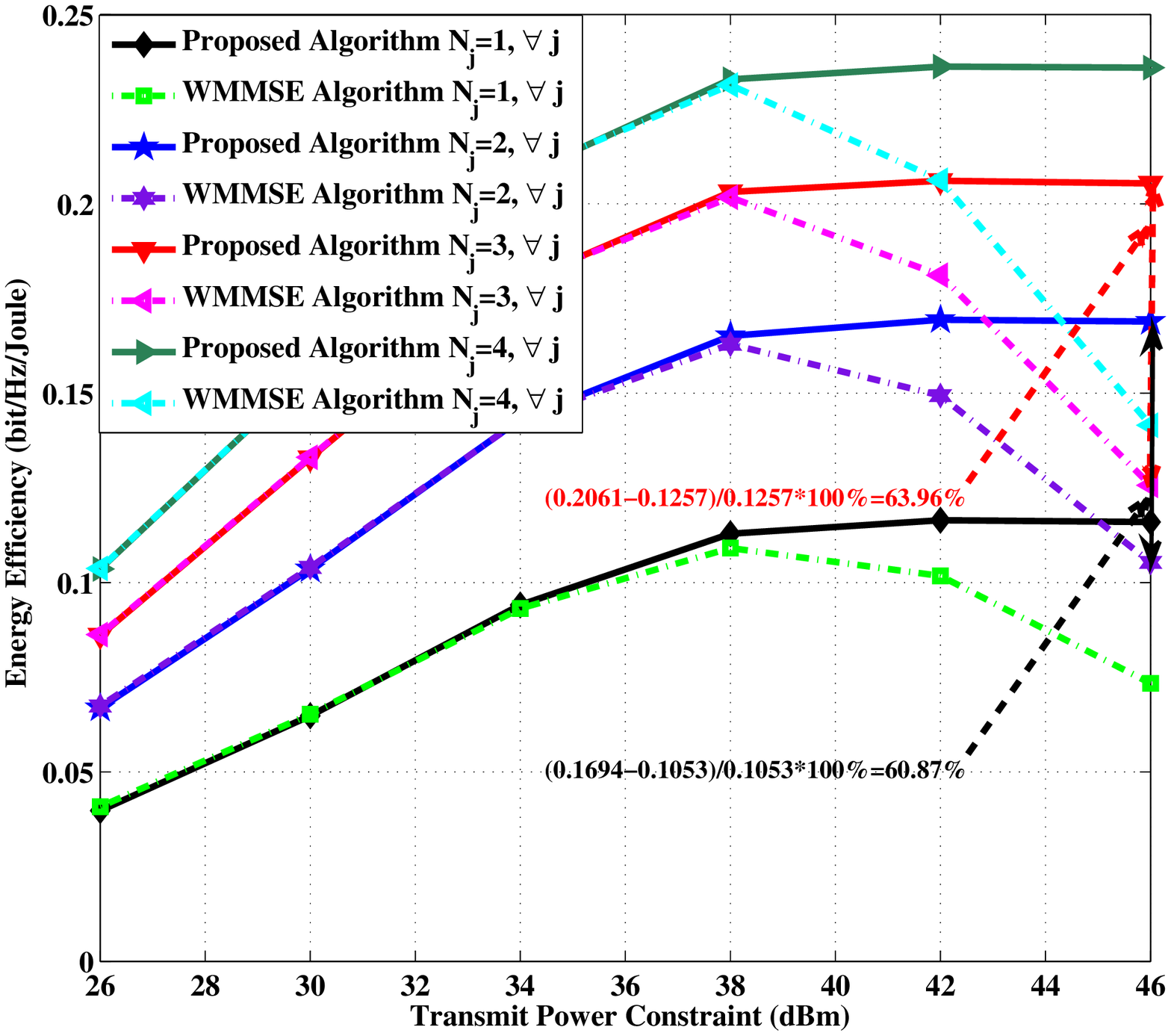}\\
\caption{Energy efficiency comparison Vs transmit power constraint, $M_{j}=4$, $\forall j$.}
\label{MultiUsersE}
\end{figure}

Fig.~\ref{MultiMethods} illustrates the average energy efficiency of the proposed algorithm, the WMMSE algorithm and the power allocation algorithm with fixed transmit beamformers,  over $10000$ random channel realizations under a scenario where each user has at least $415$ m distance from its serving BS. Numerical results show that the proposed algorithm and the WMMSE algorithm achieve obvious performance gain than the power allocation algorithm in terms of energy efficiency. This implies that the transmit beamforming optimization plays a key role in our proposed algorithm and the WMMSE algorithm. In other words, transmit beamforming vectors and the transmit power allocation should be jointly optimized for designing an energy efficient transmission design. In addition, numerical results also corroborate that the energy efficiency of the power allocation algorithm is saturated at the high transmit power region. Comparing with the observations in Fig.~\ref{MultiUsersE}, it is shown that though the energy efficiency performances of both our algorithm and the WMMSE algorithm decrease when the users move to the cell edge, the gain of the proposed algorithm over the WMMSE algorithm become more obvious.

\begin{figure}[h]
\centering
\captionstyle{flushleft}
\onelinecaptionstrue
\includegraphics[width=0.8\columnwidth,keepaspectratio]{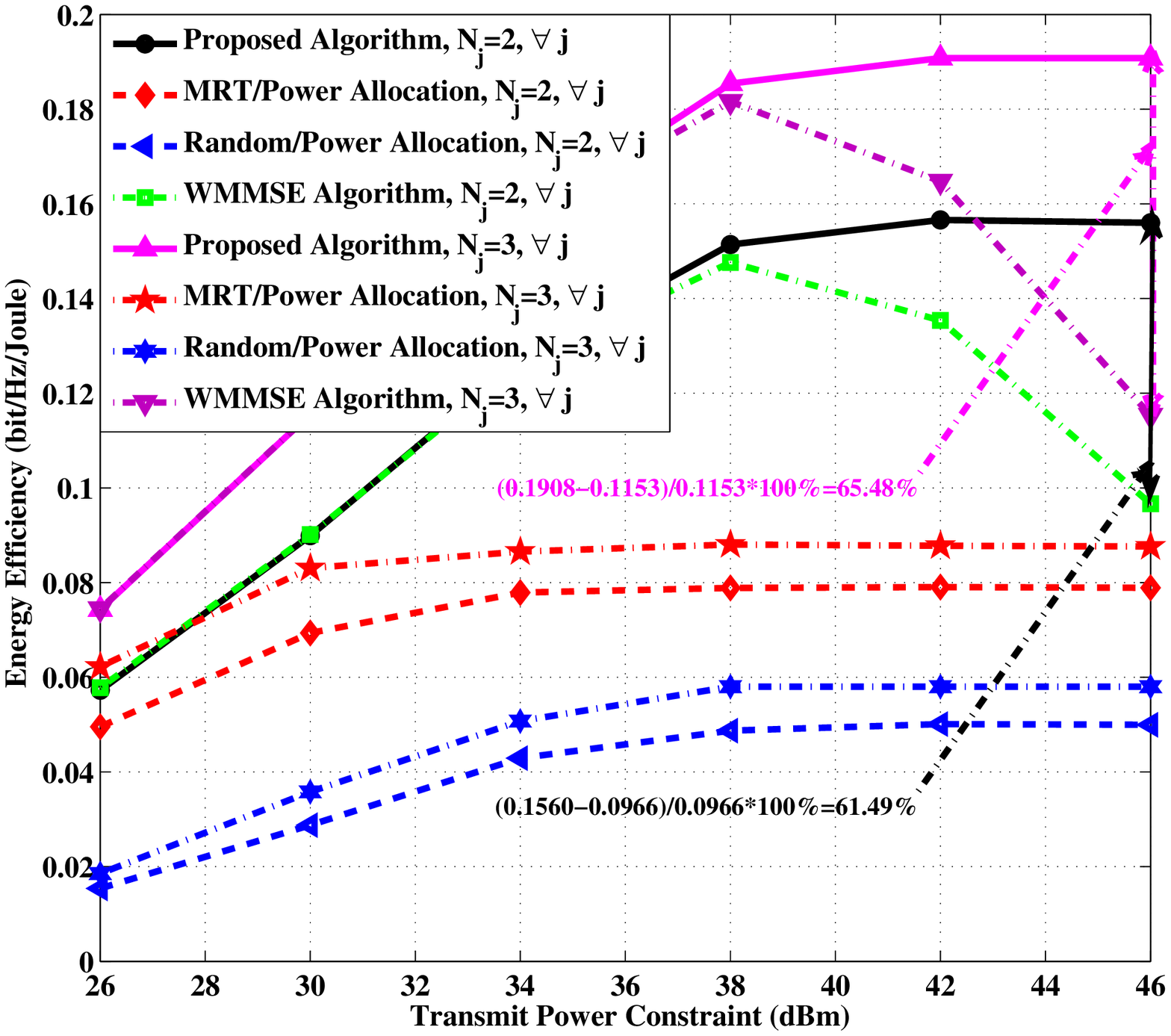}\\
\caption{Energy efficiency comparison Vs transmit power constraint, $M_{j}=4$, $\forall j$.}
\label{MultiMethods}
\end{figure}

Fig.~\ref{Smallcell} illustrates the average energy efficiency of the proposed algorithm and the WMMSE algorithm in a small cell environment where the cell radius is set to $100$ m and each user has at least $70$ m distance from its serving BS. In this case the large scale fading factor in decibels is given as $10\log_{10}(\theta_{m,j,k})=-30\log_{10}(d_{m,j,k})-38+\eta_{m,j,k}$~\cite{3GPP}. Numerical results show that the proposed algorithm outperforms the WMMSE algorithm even at the low transmit power region from $20$ to $30$ dBm. Note that at this region in a normal size cell environment, as shown in Fig.~\ref{MultiUsersE},  these two algorithms almost achieve the same energy efficiency performance.

\begin{figure}[h]
\centering
\captionstyle{flushleft}
\onelinecaptionstrue
\includegraphics[width=0.8\columnwidth,keepaspectratio]{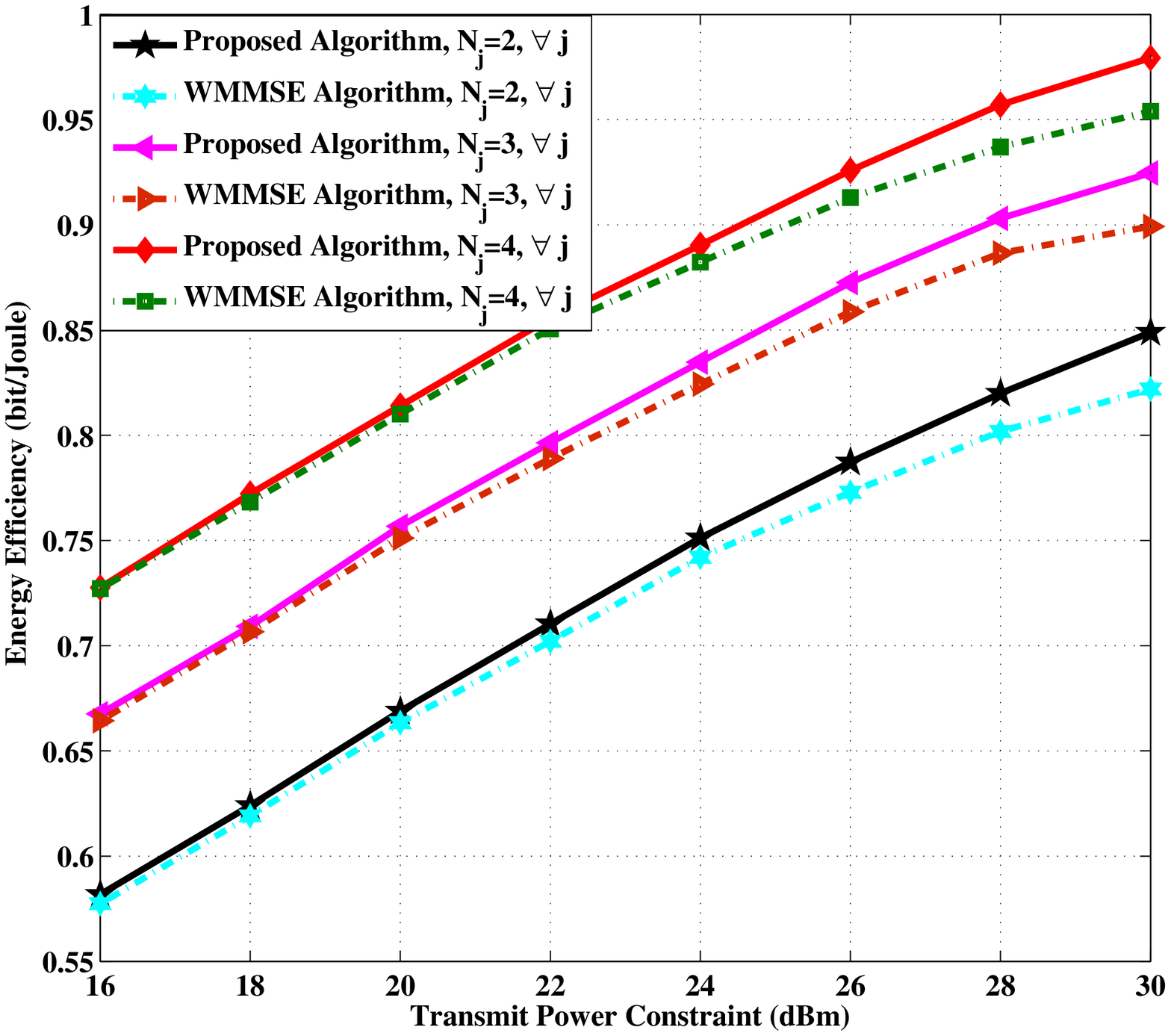}\\
\caption{Energy efficiency comparison Vs transmit power constraint, $M_{j}=4$, $\forall j$.}
\label{Smallcell}
\end{figure}

Fig.~\ref{AntennasPowers} shows that the average energy efficiency of the above two algorithms varying with the number of transmit antennas over $2000$ random channel realizations, where the number of the served users at each BS is configured to increase with the number of transmit antennas according to a fixed ratio which is set to $1:4$ in our simulations. Numerical results show that the the energy efficiency of these two algorithms both increase with the number of transmit antennas. The proposed algorithm exhibits obvious advantage over the WMMSE algorithm at $P_{j}=46$ dBm especially when the number of transmit antennas is not so large. It is also observed that the gain of the proposed algorithm over the WMMSE tends to shrink with the number of transmit antennas, implying that the sum rate optimal scheme shows high energy efficiency in the large scale MIMO system.

\begin{figure}[h]
\centering
\captionstyle{flushleft}
\onelinecaptionstrue
\includegraphics[width=0.8\columnwidth,keepaspectratio]{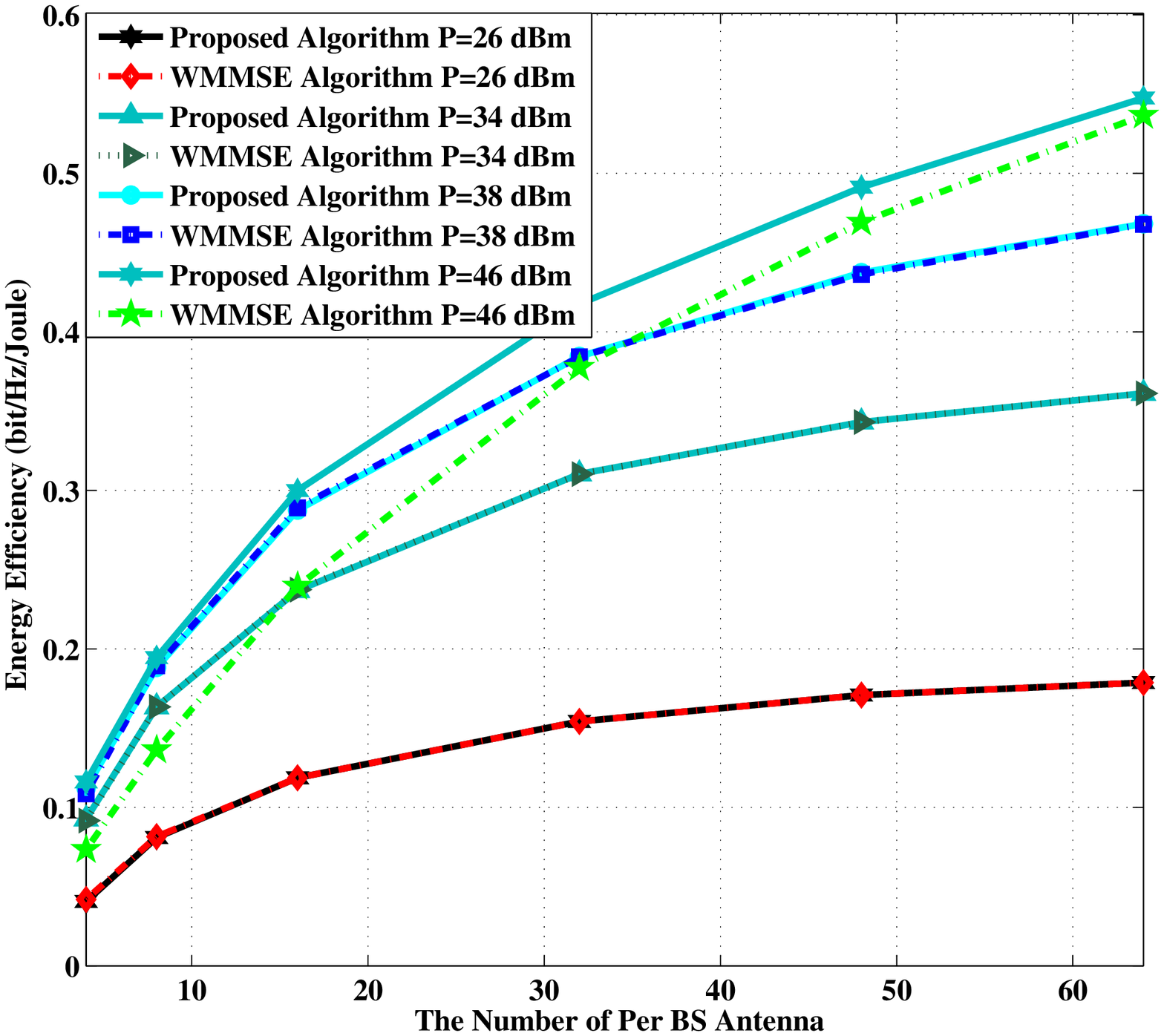}\\
\caption{Energy efficiency comparison Vs the number of per-BS Antenna.}
\label{AntennasPowers}
\end{figure}

Fig.~\ref{EEVSConstants} illustrates the average energy efficiency of the above two algorithms for different values of circuit power per antenna over $10000$ random channel realizations. It can be seen that the energy efficiency performance of both two algorithms improves with the circuit power per antenna decreasing. If $P_{c}$ can be reduced from 40dBm to 30dBm, the energy efficiency of the proposed algorithm improves by over 100\% gain. Moreover, it is also observed that the upper bound of the transmit power region where the proposed algorithm and the WMMSE algorithm achieve the same energy efficiency performance moves from 42dBm to 38dBm, implying that the proposed algorithm achieves more energy efficiency gain over the sum rate maximization algorithm if the circuit power is reduced.

\begin{figure}[h]
\centering
\captionstyle{flushleft}
\onelinecaptionstrue
\includegraphics[width=0.8\columnwidth,keepaspectratio]{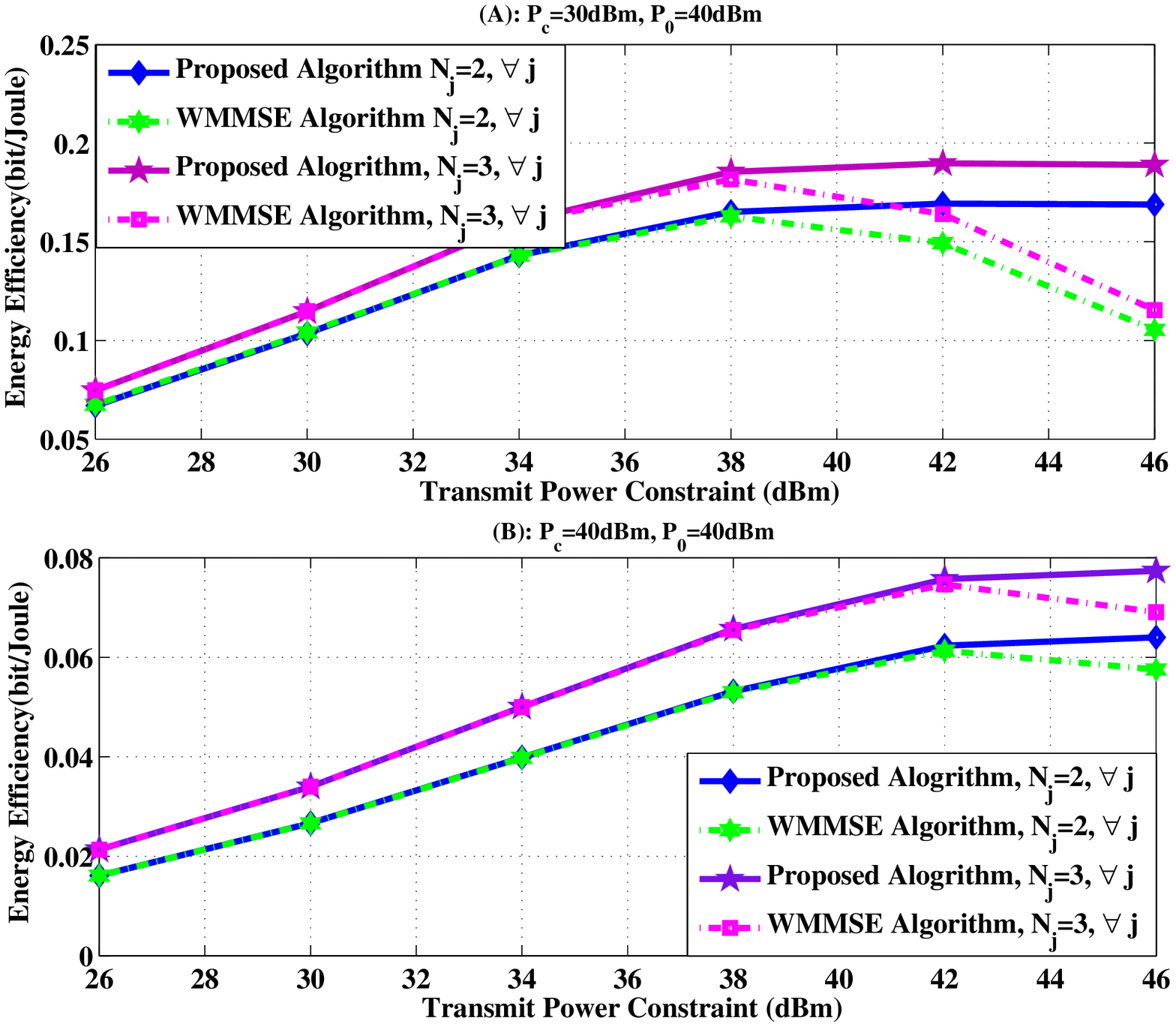}\\
\caption{Energy efficiency comparison Vs Different circuit power per antenna, $M_{j}=4$, $\forall j$.}
\label{EEVSConstants}
\end{figure}

Fig.~\ref{SRMultiUsers} shows the average sum rate of the proposed algorithm and the WMMSE algorithm over $10000$ random channel realizations. One can see that the proposed algorithm achieves the same sum rate as the WMMSE algorithm at the low transmit power region. While at the high transmit power region, the average sum rate of the proposed algorithm becomes saturated. This is because the proposed algorithm tends to reduce the transmit power at this region in order to maximize the system energy efficiency, while the WMMSE algorithm targets at the maximum sum rate which usually always transmits at a full power.

\begin{figure}[h]
\centering
\captionstyle{flushleft}
\onelinecaptionstrue
\includegraphics[width=0.8\columnwidth,keepaspectratio]{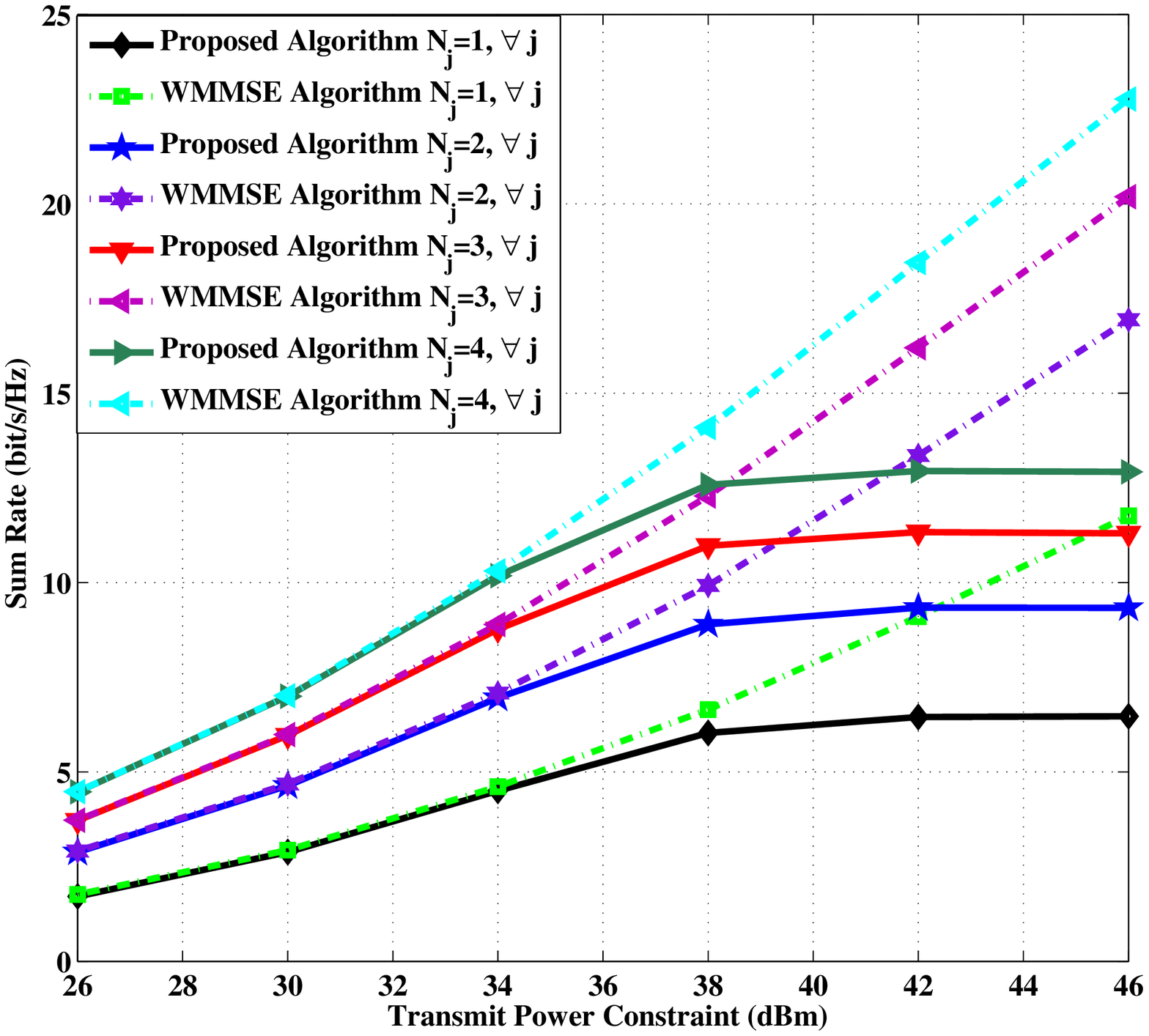}\\
\caption{Sum rate comparison Vs transmit power constraint, $M_{j}=4$, $\forall j$.}
\label{SRMultiUsers}
\end{figure}

\section*{\sc \uppercase\expandafter{\romannumeral6}. Conclusions}

In this paper, we studied the energy efficient coordinated beamforming and power allocation for the multicell multiuser downlink system. The original optimization problem is non-convex and in a fractional form. To solve it, using fractional programming, the original optimization problem was transformed into an equivalent subtractive form problem. An efficient optimization algorithm was then developed to find a solution to the equivalent optimization problem. The convergence of the proposed algorithm was proved and the solution was further derived in closed form. Numerical results illustrated that the proposed algorithm always converges to a stable point within a limited number of iterations and achieved a near-optimal performance. In particular, it was observed that at the low transmit power region the proposed algorithm obtained both the near-optimal energy efficiency and the near-optimal sum rate at the same time. Interesting topics for future work include studying the impact of imperfect CSI at each transmitter on the energy efficiency and finding the global optimum to the sub-problem for each given energy efficiency factor $\eta$.

\begin{small}

\end{small}
\end{document}